\documentclass[letterpaper,11pt,oneside,onecolumn]{article}
\usepackage{amsmath, amssymb, amstext}
\usepackage{amsthm}
\usepackage{csquotes}

\usepackage{mathtools}
\usepackage{authblk}
\usepackage[colorlinks]{hyperref}
\usepackage{enumerate}

 \newcommand{\R}{\mathbb{R}}
\newcommand{\F}{\mathbb{F}} 
\newcommand{\Z}{\mathbb{Z}}

\newcommand{\0}{\mathbf{0}} 
\newcommand{\1}{\mathbf{1}} 
 
\renewcommand{\L}{\mathcal{L}^\star}
\renewcommand{\P}{\mathcal{P}^\star}

\newcommand{\fix}{\mathrm{Fix}}
\DeclareMathOperator{\spanOp}{span}
\newcommand{\lspan}{\spanOp}
\DeclareMathOperator{\lexOp}{lex}
\newcommand{\lex}{\lexOp}

\newcommand{\ip}[2]{\langle{#1},{#2}\rangle}

\newtheorem{theorem}{Theorem}

\newtheorem{lemma}[theorem]{Lemma}
\newtheorem{remark}[theorem]{Remark}

\author{Kanstantsin Pashkovich}

\title{Computing the Nucleolus of Weighted Voting Games in Pseudo-polynomial Time}

\affil{Department of Combinatorics and Optimization,\\ University of Waterloo,\\ 200 University Avenue West,
Waterloo, ON, Canada N2L 3G1\\
  \texttt{kpashkov@uwaterloo.ca}}

\begin{document}

\maketitle 

\begin{abstract}
  We provide an algorithm for computing the nucleolus for an  instance of a weighted voting game in pseudo-polynomial time. This resolves an open question posed by Elkind.~et.al. 2007. 
\end{abstract}

\section{Introduction}

The cooperative game theory studies the situation when individuals start to collaborate in order to achieve their goals. An instance of a cooperative game is given by a set $N$ of individuals, so called \emph{players}, and a \emph{utility function}, which assigns a value to every possible coalition of players. 

A main question of the cooperative game theory asks for a distribution of the utility value of the \emph{grand coalition} $N$ among all players so that the players are encouraged to cooperate, i.e. so that the incentives of players to leave the grand coalition are minimized. A natural way to address this question is to distribute  the utility value of the grand coalition in a ``fair" way among the players. \emph{Nucleolus} is the concept introduced by Schmeidler in~\cite{Schmeidler69}  and it can be seen as such a ``fair" and desirable distribution. It is not surprising that distributions corresponding to the nucleolus appeared long before the formal definition by Schmeidler and date back to Babylonian Talmud~\cite{Aumann85}. The nucleolus was extensively studied for different types of cooperative games: flow games~\cite{Xiaotie09}, shortest path games~\cite{Barahona17}, assignment games~\cite{Nunez04}, matching games~\cite{Kern03},~\cite{Pashkovich18},~\cite{Paulusma01}, neighbor games~\cite{Herbert03}, min-cost spanning tree games~\cite{Faigle98} etc. For some of these games the nucleolus can be computed efficiently, for some the nucleolus problem is $\mathcal {NP}$-hard.

In this paper, we study the problem of computing the nucleolus for \emph{weighted voting games}. In a weighted voting game with $n$ players, we have a non-negative integral \emph{threshold} $W$ and for each player we have an associated non-negative integral \emph{weight} $w_i$, $i=1,\ldots,n$. The weight $w_i$ of the player~$i$ indicates the contribution of this player to a coalition. The utility value~$\nu(S)$ of a player coalition $S\subseteq N$ equals $1$ if and only if $w(S)\geq W$, otherwise the utility value $\nu(S)$ equals $0$. Here, we use the notation $w(S):=\sum_{i\in S} w_i$. 

Weighted voting games model decision making processes in political bodies, where the players correspond to different parties and their weight corresponds to the number of delegates from these parties. Once the number of delegates voting in favour of some decision is larger than a given threshold, the decision is going to be implemented.

\subsection{Our Result}

Weighted voting games are notoriously hard for the classical concepts of ``fair" distributions, namely it is $\mathcal{NP}$-hard to compute a least core allocation and the nucleolus~\cite{Elkind2009}. On the positive note, the paper~\cite{Elkind2009} provides an algorithm to compute a least core allocation in pseudo-polynomial time. The question about a similar result for the nucleolus was stated as an open question in~\cite{Elkind2009}. The current paper answers this question positively. Moreover, our pseudo-polynomial running time algorithm can be generalized in a straightforward manner to games corresponding to the intersection of a constant number of weighted voting games. 

\subsection{History of the Question}
The question about an algorithm for computing the nucleolus of a weighted voting game in pseudo-polynomial time was posed in~\cite{Elkind2009}. Later in~\cite{Pasechnik2009}(SODA 2009)  a novel algorithm  appeared. It was claimed that  the algorithm in~\cite{Pasechnik2009} computes the nucleolus of an instance of a weighted voting game and that this algorithm runs in pseudo-polynomial time. The algorithm in~\cite{Pasechnik2009} is based on beautiful geometric and discrete ideas. Unfortunately, the algorithm in~\cite{Pasechnik2009} is not guaranteed to compute the nucleolus. Indeed, it is based on a sequence of linear programs, so that the nucleolus is an optimal solution of each of the linear programs. However, it is not known whether the final linear program in~\cite{Pasechnik2009}  has a unique optimal solution, and so it is not known whether the algorithm in~\cite{Pasechnik2009}  outputs the nucleolus. The fact that the algorithm in~\cite{Pasechnik2009} may lead to a wrong output was reported by Solymosi in Mathematical Reviews~\cite{Solymosi2009}. 
\begin{displayquote}
\emph{``\ldots Despite all of these noteworthy technical contributions to the computation of the nucleolus in general, and for weighted voting games in particular, the main claim that the presented procedure is a pseudopolynomial time algorithm is not justified. \ldots it seems that the pseudopolynomial computability of the nucleolus for weighted voting games remains an open problem."}

\smallskip
\emph{Tam\'{a}s Solymosi}
\end{displayquote}

\subsection{Our Approach}
To find the nucleolus we solve a sequence of linear programs. Our sequence of linear programs does not directly correspond to the well-known Maschler's scheme~\cite{Maschler79}, bur rather relaxes it. Even though in our scheme we might need to solve more linear programs than in the Maschler's scheme, the total number of linear programs in our scheme is bounded by the number of players. The final linear program in our scheme has a unique optimal solution, and this optimal solution is the nucleolus. 

 To solve the linear programs arising in our scheme, we use a dynamic programming approach inspired by~\cite{Pasechnik2009}.  However, to implement our scheme and to resolve the problem with the algorithm from~\cite{Pasechnik2009}, we need to incorporate ``linear independence" constraints in dynamic programs. A naive approach would not lead to dynamic programs of pseudo-polynomial size. We resolve this issue by substituting  a ``linear independence" constraint over $\R$ with several ``linear independence" constraints over fields $\F_p$, where $p$ ranges over a selected set of prime numbers such that the unitary encoding of each considered prime number $p$ is of  polynomial size in the number of players.

\section{Preliminaries}
A \emph{cooperative game} $(N,\nu)$ is defined by a set of players $N$ and a \emph{utility function} $\nu:2^N\rightarrow \R_+$.
A weighted voting game $(N,\nu)$ can be defined by a set of players $N$, a non-negative integral vector $w\in \Z_+^N$ and a non-negative integral threshold $W\in \Z_+$, where for each coalition of players $S\subseteq N$ we have
\[
\nu(S):=\begin{cases} 1 & \text{if}\quad w(S)\geq W\\
0& \text{otherwise}
\end{cases}\,.
\]

A central question in cooperative game theory is asking for a fair distribution of $\nu(N)$, i.e. a  distribution of the value of the \emph{grand coalition} among all players in a most ``fair" way. Below, we give one possible mathematical interpretation of this question.

For a cooperative game $(N,\nu)$, a vector $x\in \R^N$ is called an \emph{allocation} if $x\geq \0$ and $x(N)=\nu(N)$. Given an allocation $x\in \R^N$, we define the \emph{excess} of a coalition $S$ with respect to the allocation $x$ to be the value $x(S)-\nu(S)$. Intuitively, the smaller is the excess of a coalition $S$ the more  the coalition $S$ is inclined to leave the grand coalition. The next linear program reflects the goal of finding an allocation $x$, which maximizes the smallest excess with respect to $x$.
\begin{align*} \label{eq:coalition_leastcore}
\max ~~ & \varepsilon \tag{$P$} \\
\text{s.t.} ~~ & x(S) \geq \nu(S) + \varepsilon \qquad \text{for
                          all}\quad S \subseteq N \\
&  x(N) = \nu(N)\\
& x \geq \0.
\end{align*}
Let $\varepsilon^\star$ be the optimal value of the linear program~\eqref{eq:coalition_leastcore}. Let $P(\varepsilon^\star)$ be the set of allocations $x\in \R^N$ such that $(x,\varepsilon^\star)$ is feasible for the linear program~\eqref{eq:coalition_leastcore}. We call $P(\varepsilon^\star)$ the \emph{leastcore} of the game $(N,\nu)$.  More generally, given a linear program $Q$ in variables $\R^N\times \R$ let us denote by $Q(\varepsilon)$ the set of value assignments $x$ such that $(x,\varepsilon)$ is feasible for $Q$.

The \emph{excess vector} of an allocation $x\in \R^N$ is the vector 
\[
\Theta(x):=\big(x(S_1)-\nu(S_1),\,x(S_2)-\nu(S_2),\ldots,x(S_{2^{|N|}})-\nu(S_{2^{|N|}})\big)\,,
\]
where $\{S_1,\ldots, S_{2^{|N|}}\}$ is the list of all possible player coalitions so that
\[
x(S_1)-\nu(S_1)\leq x(S_2)-\nu(S_2)\leq\ldots\leq x(S_{2^{|N|}})-\nu(S_{2^{|N|}})\,.
\]
Thus, the excess vector lists excess values of all player coalitions, starting from the smallest excess and ending with the largest excess with respect to the given allocation $x$.

The \emph{nucleolus} is the allocation achieving the lexicographically maximum excess vector, i.e. the nucleolus can be defined as follows
\[
\arg \lex \max \{\Theta(x)\,:\,x\in \R^N,\, x(N)=\nu(N),\, x\geq \0\}\,.
\]
The concept of nucleolus was introduced in~\cite{Schmeidler69}. In~\cite{Schmeidler69}  it was shown that the nucleolus is well defined and is unique.

Clearly, allocations in the leastcore $P(\varepsilon^\star)$ maximize the smallest excess. However, generally leastcore allocations do not attempt to maximize other excesses except the smallest one. In contrast to the leastcore allocations, the nucleolus takes into account excess values of all player coalitions: first maximizing the smallest excess of a player coalition, then maximizing the second smallest excess and so on.

\begin{theorem}[Main Theorem]\label{th:main-result}
The nucleolus of a weighted voting game with $n$ players, weights $w_i\in \Z_+$, $i=1,\ldots, n$ and threshold $W\in \Z_+$ can be computed in time polynomial in $n$, $w_i$, $i=1,\ldots, n$ and $W$. 
\end{theorem}

The proof of Theorem~\ref{th:main-result} is based on a sequence of linear program, which we describe in Section~\ref{sec:maschler} and which constitutes a relaxation of the well-known Maschler's scheme.   In Section~\ref{sec:algorithm}, we explain how for an instance of a weighted voting game our scheme can be executed in time polynomial in $n$, $w_i$, $i=1,\ldots, n$ and $W$.

\section{Our Scheme for Cooperative Games}\label{sec:maschler}

Our algorithm is based on an efficient way to solve a sequence of linear programs. As the Maschler's scheme, our scheme requires us to solve a polynomial number of linear programs, where each linear program is of (potentially) exponential size. In this section, we provide a detailed overview of how the scheme works. We would like to note that the scheme, stated in this section, leads to the nucleolus of any cooperative game~$(N,\nu)$.

Given a cooperative game $(N,\nu)$, let us define $(P_1)$ to be the linear program~\eqref{eq:coalition_leastcore} and $\varepsilon_1$ to be the optimal value of~$(P_1)$, i.e. let us define $\varepsilon_1$ to be $\varepsilon^\star$. For $j=2,\ldots, n$ let us introduce the following linear program with the optimal value $\varepsilon_{j}$
\begin{align*} \label{eq:mashler_lp}
\max ~~ & \varepsilon \tag{$P_j$} \\
\text{s.t.} ~~ & x(S) \geq \nu(S) + \varepsilon \qquad \text{for
                          all}\quad S \subseteq N,\quad \chi(S) \not\in  \L_{j-1}\\                        
&  x\in P_{j-1}(\varepsilon_{j-1})\,,
\end{align*}
where we define
\[
\L_{1}:= \lspan\{\chi(N)\}
\]
and for $j=2,\ldots,n$, we define $\L_{j}$ to be a linear subspace of
\[
 \lspan\{\chi(S)\,:\,S\in \fix(P_{j}(\varepsilon_{j}))\}
\]
such that $\L_{j}\supseteq \L_{j-1}$ and $\dim(\L_{j})=j$. Here, $\fix(P_{j}(\varepsilon_{j}))$, $j=1,\ldots,n$ is defined as follows
\[
\fix(P_{j}(\varepsilon_{j})):=\{S\subseteq N\,:\, x(S)=x'(S)\quad\text{for all}\quad x,x'\in  P_{j}(\varepsilon_{j})\}\,.
\] 
Note, that
\[
\dim (\L_{1})=1 \qquad\text{and}\qquad \L_{1}\subseteq \lspan\{\chi(S)\,:\,S\in \fix(P_{1}(\varepsilon_{1}))\}\,.
\]

\begin{lemma}\label{lem:well_defined_scheme}
For $j=1,\ldots, n$,  the linear program \eqref{eq:mashler_lp} is well defined.
\end{lemma}
\begin{proof}
Clearly, for $j=1$ the linear program \eqref{eq:mashler_lp} is well defined. To prove that the linear program \eqref{eq:mashler_lp} is well defined for $j=2,\ldots,n$, it is enough to show that there exists a choice of $\L_j$ for $j=2,\ldots,n$, satisfying 
\begin{equation}\label{eq:well_defined_scheme}
\begin{aligned}
 &\L_j \subseteq \lspan\{\chi(S)\,:\,S\in \fix(P_{j}(\varepsilon_{j}))\},\\
 & \L_{j}\supseteq \L_{j-1} \qquad \text{and}\qquad \dim(\L_{j})=j\,.
\end{aligned}
\end{equation}
Let us prove this by induction. For $j=1$, the above statements trivially hold by the choice of $\L_1$ and $(P_1)$. Let us assume that we have $\L_j$ satisfying the statements~\eqref{eq:well_defined_scheme} for $j=1,\ldots,k$, $k<n$. The linear program $P_{k+1}$ is not unbounded, because $\dim(\L_{k})=k<n$  and the linear program $P_{k+1}$ implicitly includes constraints $x(N)=\nu(N)$, $x\geq \0$. Hence, there exists $S_{k+1}\subseteq N$,  $\chi(S_{k+1}) \not\in  \L_{k}$ such that for all allocations $x$ in $P_{k+1}(\varepsilon_{k+1})$ we have $x(S_{k+1})=\nu(S_{k+1})+\varepsilon_{k+1}$. Thus, we have
\[
\chi(S_{k+1})\not\in \L_{k}\qquad \text{and}\qquad S_{k+1}\in \fix(P_{k+1}(\varepsilon_{k+1}))\,.
\]
Now, it  is clear that $\L_{k+1}:=\lspan\{\chi(S_{k+1})\cup \L_{k}\}$ satisfies the necessary conditions~\eqref{eq:well_defined_scheme}, finishing the proof.
\end{proof}

\begin{lemma}
The nucleolus of a cooperative game $(N,\nu)$ equals  $P_n(\varepsilon_n)$.
\end{lemma}
\begin{proof}
Clearly, $\dim(\L_n)=n$ and $\L_n\subseteq  \lspan\{\chi(S)\,:\,S\in \fix(P_n(\varepsilon_n))\}$. By the definition of $\fix(P_n(\varepsilon_n))$, we have that $P_n(\varepsilon_n)$ consists of a single point. It remains to show that the nucleolus of $(N,\nu)$ coincides with the unique point in $P_n(\varepsilon_n)$. This is due to the fact that the nucleolus of $(N,\nu)$ lies in $P_j(\varepsilon_j)$ for all $j=1,\ldots,n$.
\end{proof}

\section{Our Scheme for Weighted Voting Games}

In our algorithm, we solve the separation problem  for the linear programs~\eqref{eq:mashler_lp}, $j=1,\ldots,n$ implicitly. Namely, in the case of weighted voting games we reformulate the linear programs~\eqref{eq:mashler_lp}, $j=1,\ldots,n$ and then solve the separation problem for our reformulation in an efficient way. To reformulate the linear programs~\eqref{eq:mashler_lp}, we need the following lemma. 

\begin{lemma}\label{lem:lin_dependence}
Given a set of different prime numbers $\P$,  $|\P|\geq \log_2(n!)$, $n\geq 3$ and a set of $0/1$-vectors $z_1, \ldots, z_k\in \{0,1\}^n$, we have that $z_1, \ldots, z_k$ are linearly independent over $\R$ if and only if $z_1, \ldots, z_k$ are linearly independent over $\F_p$ for some $p\in \P$.
\end{lemma}
\begin{proof}
Clearly, if $z_1, \ldots, z_k$ are linearly independent over $\F_p$ for some $p\in \P$ then $z_1, \ldots, z_k$ are linearly independent over $\R$. Let us show the other implication.

Let us assume that $z_1, \ldots, z_k$ are linearly independent over $\R$ but are linearly dependent over $\F_p$ for all $p\in \P$. Thus, $k\leq n$ and there exists a $k\times k$ minor $\alpha$, $\alpha\neq 0$ in the matrix with the columns defined by  $z_1, \ldots, z_k$. On one side, $\alpha$ is an integer and $|\alpha|\leq k!\leq n!$, since the vectors $z_1, \ldots, z_k$ are $0/1$-vectors. On the other side, $\alpha\equiv 0 \pmod p$ for all $p\in \P$, since $z_1, \ldots, z_k$ are linearly dependent over $\F_p$ for all $p\in \P$. Hence, we have
\[
n!< \prod_{p\in \P}p\leq |\alpha|\leq n!\,,
\]
where the first inequality is due to the fact $|\P|\geq \log_2(n!)$, leading to a contradiction and finishing the proof.
\end{proof}

Now, let us assume that we are given a set of different prime numbers $\P$,  $|\P|\geq \log_2(n!)$ and $n\geq 3$.
Using Lemma~\ref{lem:lin_dependence}, we can reformulate~\eqref{eq:mashler_lp}, $j=2,\ldots,n$ as follows 
\begin{align*} \label{eq:mashler_lp_prime}
\max ~~ & \varepsilon \tag{$\tilde{P}_j$} \\
\text{s.t.} ~~ & x(S) \geq \nu(S) + \varepsilon \qquad \text{for
                          all}\quad S \subseteq N\\&\qquad\text{such that}\,  z^{j-1}_1, \ldots, z^{j-1}_{j-1}, \chi(S)\,\text{are}\\&\qquad\text{linearly independent over}\,\F_p\,\text{for some}\, p\in \P  \\                        
&  x\in P_{j-1}(\varepsilon_{j-1})\,,
\end{align*}
where for $j=1,\ldots, n$, $z^j_1, \ldots, z^j_j\in\{0,1\}^n$  are such that 
\[
\L_j=\lspan\{z^j_1, \ldots, z^j_j\}\,.
\]
Please see Section~\ref{sec:alg_encode} for further details on the vectors $z^j_1, \ldots, z^j_{j}\in\{0,1\}^n$ for $j=1,\ldots, n$.

However, the reformulations~\eqref{eq:mashler_lp_prime}, $j=2,\ldots,n$  are not enough for our purposes. For $p\in \P$ and $j=1,\ldots,n-1$, let us denote by $v^j_{p,1}, \ldots, v^j_{p,n-j}\in \F_p^n$ a basis of the linear space over $\F_p$, which is orthogonal to the linear space spanned by $z^j_1, \ldots, z^j_j$ over  $\F_p$. Note, that in the case when $z^j_1, \ldots, z^j_j\in \{0,1\}^n$  are linearly dependent over $\F_p$ for $p\in \P$ and $j=1,\ldots,n-1$, we let $v^j_{p,1}:=\0, \ldots, v^j_{p,n-j}:=\0$.

Then we can reformulate~\eqref{eq:mashler_lp_prime}, $j=2,\ldots,n$ as follows 
\begin{align*} \label{eq:mashler_lp_final}
\max ~~ & \varepsilon \tag{$\bar{P}_j$} \\
\text{s.t.} ~~ & x(S) \geq \nu(S) + \varepsilon \qquad \text{for
                          all}\quad S \subseteq N\\&\qquad\text{such that}\quad\ip{v^{j-1}_{p,t}}{\chi(S)}\not\equiv 0\pmod p\\&\qquad\text{for some}\, p\in \P\,\text{ and }\, t=1,\ldots, n-(j-1)\\                        
&  x\in P_{j-1}(\varepsilon_{j-1})\,,
\end{align*}
where $\ip{v}{u}$ stands for the scalar product of two vectors $v$ and $u$. For the sake of exposition, we let $(\bar{P}_1)$ to be the linear program $(P_1)$.

\bigskip

In Section~\ref{sec:algorithm}, we explain how to find a desired set $\P$ and compute the vectors $v^j_{p,1}, \ldots, v^j_{p,n-j}\in \F_p^n$  for $j=1,\ldots, n$ and $p\in \P$.

\section{Algorithm}\label{sec:algorithm}

To prove Theorem~\ref{th:main-result}, we show how to solve the linear programs ~\eqref{eq:mashler_lp_final}, $j=1,\ldots,n$ in time polynomial in $n$, $w_i$, $i=1,\ldots, n$ and $W$. 

\subsection{Finding Set of Prime Numbers $\P$}\label{sec:alg_prime_numbers}

Let us compute the set $\P$, $|\P|\geq\log_2(n!)$ of different prime numbers in polynomial time, such that the size of the unitary encoding for each number in $\P$ is polynomial in $n$. A naive way of finding such set $\P$ works: inspect numbers one by one starting from $2$ until the required set of  $\lceil\log_2(n!)\rceil$ smallest prime numbers is found. Indeed, the Prime Number Theorem guarantees that such a naive search for a set $\P$ finishes in time polynomial in~$n$. Thus we get the following remark.

\begin{remark}\label{rem:alg_prime_numbers}
A set $\P$ of at least  $\log_2(n!)$ different prime numbers, such that the size of the unitary encoding for each number in $\P$ is polynomial in $n$, exists  and can be found in time polynomial in $n$. 
\end{remark}

\subsection{Encoding Linear Spaces $\L_j$, $j=1,\ldots,n$ }\label{sec:alg_encode}

In our algorithm, a linear space $\L_j$, $j=1,\ldots,n$ is a linear space spanned by $0/1$-vectors. We encode  $\L_j$, $j=1,\ldots,n$ as the set $$z^j_1, \ldots, z^j_j\in\{0,1\}^n$$ such that $\L_j=\lspan\{z^j_1, \ldots, z^j_j\}$. For simplicity of exposition, we define $\L_0:=\{\0\}$ and for it we use the trivial encoding. For $\L_1=\lspan\{\1\}$ we use the encoding $z^1_1:=\1$. Note, the encoding size of each $\L_j$, $j=0,1,\ldots,n$ is polynomial in $n$.

\subsection{Computing Basis of Orthogonal Space over $\F_p$, $p\in \P$}

For $j=0,1,\ldots,n$ and $p\in \P$, if we are given the vectors $z^j_1, \ldots, z^j_{j}\in\{0,1\}^n$ then it is straightforward to compute a basis of the linear space orthogonal to $\lspan\{z^j_1, \ldots, z^j_{j}\}$, where all the vectors are considered over the field $\F_p$.

\begin{remark}\label{rem:alg_basis_orth}
For $j=0,1,\ldots,n$ and $p\in \P$, the vectors $v^j_{p,1}, \ldots, v^j_{p,n-j}\in \F_p^n$ can be computed in time polynomial in $n$.
\end{remark}

\subsection{Solving Separation Problem for~\eqref{eq:mashler_lp_final}, $j=1,\ldots,n$}\label{sec:alg_separation}

To show that the linear program~\eqref{eq:mashler_lp_final}, $j=1,\ldots,n$  can be solved in time polynomial in $n$, $w_i$, $i=1,\ldots,n$ and $W$ it is enough to show that the separation problem for~\eqref{eq:mashler_lp_final}, $j=1,\ldots,n$ and a point $(x,\varepsilon)$ can be solved in time polynomial in $n$, $w_i$, $i=1,\ldots,n$, $W$ and the size of the binary encoding of $(x,\varepsilon)$~\cite{Grotschel88}. 

Following the approach in~\cite{Elkind2009}, we solve the separation problem for~\eqref{eq:mashler_lp_final}, $j=1,\ldots,n$ using dynamic programming.

\begin{lemma}\label{lem:separation}
For the linear program~\eqref{eq:mashler_lp_final}, $j=1,\ldots,n$  and a point $(x,\varepsilon)$, the separation problem can be solved in time polynomial in $n$, $w_i$, $i=1,\ldots,n$, $W$ and the size of the binary encoding of $(x,\varepsilon)$. 
\end{lemma}
\begin{proof}
 Let us prove the statement of the lemma by induction on $j$. For $j=1$, we can compute the values $\gamma_{k,U}$, $k=1,\ldots,n$, $U=0,1,\ldots, w(N)$
$$
\gamma_{k, U}:=\min\{x(S)\,:\, S\subseteq \{1,\ldots,k\},\quad w(S)=U\}
$$ 
using dynamic programming in time polynomial in $n$, $w_i$, $i=1,\ldots,n$ and the size of the binary encoding of $x$. Now, to verify that $(x,\varepsilon)$ is a feasible solution for the linear program~\eqref{eq:mashler_lp_final} with $j=1$, we need to check that $\gamma_{n, U}\geq 1+\varepsilon$ for every $U\geq W$.

Now, let us assume that  we can solve the separation problem for~\eqref{eq:mashler_lp_final}, $j=1,\ldots,q$ and a point $(x,\varepsilon)$ in time polynomial in $n$, $w_i$, $i=1,\ldots,n$, $W$ and the size of the binary encoding of $(x,\varepsilon)$. Let us prove that the separation problem can be solved efficiently also for $j=q+1$ if $q<n$. For this, we solve the separation over each of the constraint families  
\begin{align*}
& x(S) \geq \nu(S) + \varepsilon \qquad \text{for
                          all}\quad S \subseteq N\\&\qquad\text{such that}\quad\ip{v^{j-1}_{p,t}}{\chi(S)}\not\equiv 0 \pmod p\,.
\end{align*}
indexed by $p\in \P$ and $t=1,\ldots,n-(j-1)$ independently.
To solve the separation problem for the above families, we can compute the values $\gamma_{k,g,U}$, $k=1,\ldots,n$, $g\in \F_p$, $U=0,1,\ldots, w(N)$
$$
\gamma_{k,g, U}:=\min\{w(S)\,:\, S\subseteq \{1,\ldots,k\},\quad v^{j-1}_{p,t}(S)\equiv g \pmod p,\quad w(S)=U\}
$$ 
using dynamic programming in time polynomial in $n$, $w_i$, $i=1,\ldots,n$ and the size of the binary encoding of $x$. Now, to verify that $(x,\varepsilon)$ is a feasible solution for the linear program~\eqref{eq:mashler_lp_final}, we need only to check that $\gamma_{n,g, U}\geq 1+\varepsilon$ for every $g\in \F_p$, $g\not\equiv 0\pmod p$ and $U\geq W$. Note, that due to the induction hypothesis we can solve the separation problem over the constraint $x\in P_{j-1}(\varepsilon_{j-1})$  in time polynomial in $n$, $w_i$, $i=1,\ldots,n$, $W$ and the size of the binary encoding of $x$, since  for this it is enough to find a constraint valid for the feasible region of $\bar{P}_{j-1}$ but violated by $(x,\varepsilon_{j-1})$.
\end{proof}

\subsection{Obtaining Encoding of Linear Space $\L_j$, $j=1,\ldots,n$ }\label{sec:update}

Note that we use Lemma~\ref{lem:separation} to solve the linear programs~\eqref{eq:mashler_lp_final}, $j=2,\ldots,n$ using ellipsoid method~\cite{Grotschel88}. Since the linear program~\eqref{eq:mashler_lp_final}  is bounded for every $j=1,\ldots,n$, one of the inequalities used in the ellipsoid method is of the form 
\[
x(S_j)\geq \nu(S_j)+\varepsilon\,,
\]
where $S_j\subseteq N$, $S_j\in \fix(P_j(\varepsilon_j))$ and $\chi(S_j)\not\in \L_{j-1}$. Such $S_j\subseteq N$ can be found by a trivial inspection of every linear inequality generated during the execution of the ellipsoid algorithm while solving~\eqref{eq:mashler_lp_final}. Note, that the condition $S_j\in \fix(P_j(\varepsilon_j))$ can be tested in time polynomial in $n$, $w_i$, $i=1,\ldots,n$ and $W$ due to Lemma~\ref{lem:separation}. Alternatively, such $S_j$ can be found using the support of an optimal solution for the dual linear program of~\eqref{eq:mashler_lp_final}, $j=2,\ldots,n$ and complementary slackness. Then, to obtain $z^j_1,\ldots,z^j_j$ from $z^{j-1}_1,\ldots,z^{j-1}_{j-1}$ we can define $z^j_1:=z^{j-1}_1$, $z^j_2:=z^{j-1}_2$, \ldots, $z^j_{j-1}:=z^{j-1}_{j-1}$ and $z^j_j:=\chi(S_j)$.

\section{Generalizations of the Result}
In a straightforward way our algorithm can be extended to compute the nucleolus of a game, corresponding to the intersection of a constant number of voting games, in pseudo-polynomial time. Moreover,  the nucleolus of a game can be computed in polynomial time, whenever this game is the intersection of $O(\log(n))$ voting games where the threshold of each voting game is bounded by a constant.

\section{Acknowledgments}
We would like to thank Dmitrii Pasechnik for pointing us to the problem of computing the nucleolus of weighted voting games in pseudo-polynomial time. We are also grateful to Jochen Koenemann and Justin Toth for helpful discussions.

\bibliography{bibliography}
\bibliographystyle{plain}

\end{document}